\documentclass[11pt]{amsart}
\usepackage[letterpaper,margin=1.08in]{geometry}
\usepackage{amssymb,mathtools}
\usepackage{microtype}
\numberwithin{equation}{section}
\renewcommand{\theequation}{\arabic{section}.\arabic{equation}}
\usepackage{enumitem}
\usepackage{cite}
\usepackage[svgnames]{xcolor}
\usepackage{hyperref}
\ifdefined\pdfinfoomitdate
\fi
\ifdefined\pdftrailerid
  \pdftrailerid{}
\fi
\ifdefined\pdfglyphtounicode
  \pdfglyphtounicode{parenleftbig}{0028}
  \pdfglyphtounicode{parenrightbig}{0029}
  \pdfglyphtounicode{bracketleftbig}{005B}
  \pdfglyphtounicode{bracketrightbig}{005D}
  \pdfglyphtounicode{parenleftBig}{0028}
  \pdfglyphtounicode{parenrightBig}{0029}
  \pdfglyphtounicode{parenleftbigg}{0028}
  \pdfglyphtounicode{parenrightbigg}{0029}
  \pdfglyphtounicode{bracketleftbigg}{005B}
  \pdfglyphtounicode{bracketrightbigg}{005D}
\fi
\hypersetup{
 colorlinks=true,allcolors=Navy,breaklinks=true,
 pdftitle={The Support-Distribution Conditions for MSRD Codes Reduce to a Single Inequality},
 pdfauthor={Xiang Fan},
 pdfsubject={Nonnegativity of formal support counts of MSRD codes},
 pdfkeywords={Elementary symmetric polynomials, maximum sum-rank distance codes, support distribution}
}

\newtheorem{theorem}{Theorem}[section]
\newtheorem{lemma}[theorem]{Lemma}
\newtheorem{proposition}[theorem]{Proposition}
\newtheorem{corollary}[theorem]{Corollary}
\theoremstyle{remark}
\newtheorem{example}[theorem]{Example}

\title[Support-distribution conditions for MSRD codes]{The Support-Distribution Conditions for MSRD Codes Reduce to a Single Inequality}
\author{Xiang Fan}
\address{School of Mathematics, Sun Yat-sen University, Guangzhou 510275, China}
\email{fanx8@mail.sysu.edu.cn}
\keywords{Elementary symmetric polynomials, maximum sum-rank distance codes, support distribution}

\begin{document}
\begin{abstract}
Byrne, Gluesing-Luerssen, and Ravagnani determined the number of
codewords with a prescribed nonzero support in a linear maximum
sum-rank distance code whose blocks have a common number of columns.
Their formula depends only on the code parameters and the support
dimensions. Evaluating it for arbitrary parameters gives a family of
necessary nonnegativity conditions for existence. We prove their
conjecture that this family is equivalent to the condition at the vector
of total dimension $d+1$ obtained by filling the longest blocks first,
where $d$ is the distance parameter. The resulting explicit criterion
applies over every finite field $\mathbb F_q$, for all row numbers
$n_i\leqslant m$ and $1\leqslant d<N$, where $m$ is the common number of
columns and $N$ is the total number of rows. It also covers the boundary
cases in which this vector fills each of its nonzero blocks completely,
not explicitly included in the conjecture's displayed parameter range.
The proof uses a generating function and an exchange identity to reduce
the problem to vectors obtained by filling the longest blocks first. We
then estimate alternating sums of elementary symmetric polynomials,
with a separate argument for $q=2$. The dual conditions admit the same reduction. For
codes that exist, the joint condition follows from sphere-packing bounds
for the code and its dual. An infinite family of parameters satisfies
the joint condition but admits no such linear code.
\end{abstract}
\maketitle

\section{Introduction}
\label{sec:intro}

Sum-rank-metric codes contain Hamming-metric and rank-metric codes as
special cases, and maximum sum-rank distance (MSRD) codes, which attain
the Singleton bound in this metric, are the sum-rank analogue of maximum
distance separable (MDS) codes; see \cite{MPSK2022} for a survey of constructions, decoding, and
applications. The weight distribution of an MDS code is determined by
its length, its dimension, and the field size
\cite[Ch.~11, Sec.~3, Theorem~6]{MS1977}. Byrne, Gluesing-Luerssen, and
Ravagnani \cite{BGR2021} proved an analogue for MSRD codes whose blocks
have a common number of columns: the number of codewords with a
prescribed nonzero support depends only on the parameters of the code
and on the dimensions of the support. Their formula can be evaluated
for arbitrary parameters, whether or not a code with these parameters
exists, and the nonnegativity of its values is a family of necessary
conditions for existence. They conjectured that the whole family is
equivalent to a single one of its inequalities
\cite[Conjecture~VI.8]{BGR2021}. We prove this conjecture. The argument
also covers some boundary values of the parameters that are not
explicitly included in its statement.

\subsection{Sum-Rank-Metric Codes}

Let $q$ be a prime power, and let $t,n_1,\ldots,n_t,m_1,\ldots,m_t$
be positive integers with $n_i\leqslant m_i$ for all
$1\leqslant i\leqslant t$. The ambient space is
$\Pi:=\bigoplus_{i=1}^t\mathbb F_q^{n_i\times m_i}$, whose elements are
tuples $X=(X_i)_{i=1}^t$ of matrices over $\mathbb F_q$. We call the
$i$-th summand, which has $n_i$ rows and $m_i$ columns, the $i$-th
\emph{block}, and $X_i$ the $i$-th block of $X$. The
\emph{sum-rank} of $X$ is
$\operatorname{srk}(X):=\sum_i\operatorname{rk}(X_i)$, where
$\operatorname{rk}$ is matrix rank, and the sum-rank distance is
$(X,Y)\mapsto\operatorname{srk}(X-Y)$. For $t=1$ this is the rank
metric, and for $n_i=m_i=1$ it is the Hamming metric.

A nonzero $\mathbb F_q$-linear code $C\subseteq\Pi$ has minimum
distance $d:=\min\{\operatorname{srk}(X):X\in C\setminus\{0\}\}$.
Deleting any $d-1$ rows across the blocks gives an injective map on
$C$, which yields the \emph{sum-rank Singleton bound}
\cite[Theorem~III.2]{BGR2021}
\begin{equation}\label{eq:singleton-general}
\dim_{\mathbb F_q}C\leqslant
\min_{\substack{a_i\in\{0,\ldots,n_i\}\\\sum_i a_i=d-1}}
\sum_{i=1}^t m_i(n_i-a_i).
\end{equation}
A nonzero $\mathbb F_q$-linear code attaining equality is an
\emph{MSRD code}. This is the linear case
of the definition \cite[Definition~III.3]{BGR2021}, which also allows
nonlinear codes. The terminology and early constructions, using linearized
Reed--Solomon codes, are due to Mart\'inez-Pe\~nas \cite{MP2018}. Those
constructions are linear over an extension field; here we only assume
$\mathbb F_q$-linearity. When every $m_i=m$, the right-hand side
of~\eqref{eq:singleton-general} equals $m(N-d+1)$, where
$N:=\sum_i n_i$.

\subsection{Support Distributions and the Conjecture}

The support of a codeword is the tuple of column spaces of its
blocks, and the \emph{support distribution} of a code records how many
codewords have each prescribed support
\cite[Definitions~II.4 and~II.8]{BGR2021}; summing over the supports of
a given total dimension recovers the weight distribution. When the
blocks have a common number of columns, the number of codewords with a
fixed nonzero support in a nonzero linear MSRD code is given by an
explicit formula \cite[Theorem~VI.5]{BGR2021}. We call the values of this
formula \emph{formal support counts}. They can
be negative for parameters that admit no code, and their nonnegativity
is a necessary condition for the existence of such a code
\cite[Corollary~VI.6]{BGR2021}.

For blocks with $n_1,\ldots,n_t$ rows, this condition consists of
$\prod_i(n_i+1)-1$ inequalities, one for each nonzero
support-dimension vector, although those of total dimension at most
$d$ hold automatically. The conjecture
\cite[Conjecture~VI.8]{BGR2021} asserts that one of them suffices,
namely the inequality at the support-dimension vector of total
dimension $d+1$ obtained by filling the longest blocks first.

\subsection{Main Result}

Order the blocks so that $m\geqslant n_1\geqslant\cdots\geqslant n_t$,
and let $1\leqslant d<N$. Filling the blocks in this order with $d+1$
units gives the vector
$\widetilde{\boldsymbol{u}}=(n_1,\ldots,n_{j-1},\delta,0,\ldots,0)$,
where $n_1+\cdots+n_{j-1}<d+1=n_1+\cdots+n_{j-1}+\delta$ and
$1\leqslant\delta\leqslant n_j$. By Theorem~\ref{thm:main}, all
formal support counts for the distance parameter $d$ are nonnegative if
and only if
\begin{equation}\label{eq:explicit-test}
 \sum_{i<j}(q^{n_i}-1)+(q^\delta-1)\leqslant(q-1)(q^m+1).
\end{equation}
This inequality is equivalent to the nonnegativity of the count at
$\widetilde{\boldsymbol{u}}$, so Theorem~\ref{thm:main} proves the
conjecture in the range in which it is stated; no code with the given
parameters is assumed to exist. The theorem also covers the values of
$d$ with $d+1=n_1+\cdots+n_j$ for some $j$, including $d+1=N$, which
are not explicitly covered by the parameter range displayed in
\cite[Conjecture~VI.8]{BGR2021}.

Write $\omega_d(\boldsymbol{u})$ for the formal support count attached to a
support-dimension vector $\boldsymbol{u}=(u_1,\ldots,u_t)$, where
$0\leqslant u_i\leqslant n_i$, and $\mathbf e_i$ for the $i$-th
standard basis vector. The proof starts from an expression of
$\omega_d(\boldsymbol{u})$ as a coefficient of a rational function
(Section~\ref{sec:reduction}). This expression yields the following
identity: if
$\lvert\boldsymbol{u}\rvert\geqslant d+1$ and a unit is moved from
coordinate $j$ of $\boldsymbol{u}$ to another coordinate $i$, where
$u_j\geqslant1$ and $u_i<n_i$, then
\[
 \omega_d(\boldsymbol{u}+\mathbf e_i-\mathbf e_j)
 -\omega_d(\boldsymbol{u})
 =-(q^{u_i}-q^{u_j-1})\,\omega_d(\boldsymbol{u}-\mathbf e_j).
\]
If $u_i\geqslant u_j$ and all counts of total dimension
$\lvert\boldsymbol{u}\rvert-1$ are nonnegative, the identity shows that
the move cannot increase the count. After sorting the coordinates,
repeated moves lead to the vector obtained by filling the longest
blocks first. Under the same nonnegativity hypothesis, this vector
therefore minimizes the count in its total dimension. Induction on the
total dimension reduces the theorem to these vectors.

After rescaling, their counts become alternating partial sums of
elementary symmetric polynomials in explicit numbers from $(0,1/q]$,
one for each filled row except the first. For $q\geqslant3$, all these
numbers are at most $1/3$, and a short induction proves nonnegativity
(Section~\ref{sec:positivity}). For $q=2$, a block with $m$ rows
contributes an entry $1/2$, so a separate argument is needed. The cases
$d\leqslant2$ require no hypothesis on~\eqref{eq:explicit-test}. For
$d\geqslant3$, Section~\ref{sec:binary} shows that, under that inequality,
an entry $1/2$ occurs beyond the first $d+1$ rows only when $n_1=m$ and
$d\leqslant m$, and treats these remaining cases.

The dual necessary conditions \cite[Corollary~VI.10]{BGR2021}
involve the same formula with another distance parameter, so
Theorem~\ref{thm:main} reduces them, and both families together, to a
single inequality (Section~\ref{sec:dual}). For linear MSRD codes that
exist, this joint condition already follows from the
projective sphere-packing bound \cite[Section~III]{BGR2021}, applied
to the code and its dual (Proposition~\ref{prop:sphere-comparison}).
The contribution of this paper is thus the equivalence for the formal
counts, not a stronger existence bound.
Example~\ref{ex:nonexistence} gives an infinite family of parameters for
which all primal and dual formal support counts are nonnegative but no
linear MSRD code exists.

Appendix~\ref{app:refinements} gives further properties of the formal
support counts. Under~\eqref{eq:explicit-test}, the counts of total
dimension at least $d+2$ are positive, with one exceptional set of
parameters. For $d\geqslant3$ and odd $r\geqslant3$, we give lower bounds
in total dimension $d+r$ and determine the equality cases for the
bound valid over every finite field. We also determine the zero counts
in total dimension $d+1$ and the minimizers in each total dimension at
least $d+2$. Without assuming~\eqref{eq:explicit-test}, we prove that
every count is a multiple of $q^m-1$ and establish congruences modulo $q$,
and modulo $q-1$ when $d\geqslant2$. Finally, we derive arithmetic
restrictions on equality in~\eqref{eq:explicit-test}.

\subsection{Related Work}

Mart\'inez-Pe\~nas developed the theory of support spaces, restriction,
and shortening for sum-rank-metric codes \cite{MP2019}, and
Camps-Moreno et al.\ studied generalized sum-rank weights of MSRD codes
\cite{CGL2022}. Xia and Fu derived support distributions for dually
quasi-MSRD codes \cite{XiaFu2026}, and Bartoli, Grimaldi, and
Longobardi constructed nonadditive MSRD codes with the same sum-rank
weight distribution as $\mathbb F_{q^m}$-linear MSRD codes with the same
parameters \cite{BGL2026}. Bounds on the sizes of sum-rank-metric codes
have been obtained from eigenvalues
\cite{AbiadMongeKhramovaRavagnani2024} and from linear programming
\cite{AbiadGavrilyukKhramovaPonomarenko2025}. The present paper
concerns instead the formal support counts themselves, and
its main theorem does not presuppose that a code with the prescribed
parameters exists.

\section{Formal Support Counts and the Main Theorem}
\label{sec:setup}

We follow the notation of Byrne et al.\ \cite[Section~VI]{BGR2021}. Throughout,
$q$ is a prime power and
\[
\begin{aligned}
&m\geqslant n_1\geqslant n_2\geqslant\cdots\geqslant n_t\geqslant1,\\
&N=\sum_{i=1}^t n_i,\qquad 1\leqslant d<N.
\end{aligned}
\]
The integer $d$ is a prescribed distance parameter; we do not assume
that a code with that minimum distance exists. All blocks have the same
number $m$ of columns; the $i$-th block has $n_i$ rows, and we call
$n_i$ its \emph{length}. The ambient space is
\[
 \Pi=\Pi_q(n_1\times m\mid\cdots\mid n_t\times m)
      =\bigoplus_{i=1}^{t}\mathbb F_q^{n_i\times m}.
\]
The \emph{support} of $X=(X_i)_{i=1}^t\in\Pi$ is the tuple of column
spaces
\[
\sigma(X):=(\operatorname{colsp}(X_1),\ldots,
\operatorname{colsp}(X_t)),
\]
where $\operatorname{colsp}(X_i)$ is the subspace of
$\mathbb F_q^{n_i}$ spanned by the columns of $X_i$
\cite[Definition~II.4]{BGR2021}. For a tuple
$\boldsymbol{U}=(U_1,\ldots,U_t)$ of subspaces
$U_i\subseteq\mathbb F_q^{n_i}$ we put
$\dim\boldsymbol{U}:=(\dim U_1,\ldots,\dim U_t)$.
By a code we always mean an $\mathbb F_q$-linear subspace
$C\subseteq\Pi$. With every $m_i=m$,
\eqref{eq:singleton-general} shows that a nonzero code with minimum
distance $d$ is MSRD precisely when
$\dim_{\mathbb F_q}C=m(N-d+1)$. A block with $n_i=0$ supports only the zero
subspace and can be omitted, so the assumption $n_i\geqslant1$ loses no
generality relative to the conjecture \cite[Conjecture~VI.8]{BGR2021}.

Let $\mathcal I:=\mathcal I_{(n_1,\ldots,n_t)}$ be the set of
support-dimension vectors, that is, of
$\boldsymbol{u}=(u_1,\ldots,u_t)\in\mathbb Z^t$ with
$0\leqslant u_i\leqslant n_i$ for every $i$, and call
$\lvert \boldsymbol{u}\rvert:=\sum_i u_i$ the \emph{total dimension}
of $\boldsymbol{u}$. We write $\mathbf e_i\in\mathbb Z^t$ for the
$i$-th standard basis vector. For integers $a\geqslant0$ and
$0\leqslant b\leqslant a$, the \emph{$q$-binomial coefficient} is
\[
 \genfrac{[}{]}{0pt}{}{a}{b}_q
 :=\prod_{j=0}^{b-1}\frac{q^{a-j}-1}{q^{b-j}-1},
\]
where the empty product is $1$; we set
$\genfrac{[}{]}{0pt}{}{a}{b}_q=0$ if $b<0$ or $b>a$. For
$0\leqslant b\leqslant a$, this number counts the $b$-dimensional
subspaces of $\mathbb F_q^a$. For $\boldsymbol{u}\in \mathcal I$ define
\begin{align*}
 f_\ell(\boldsymbol{u})
   &:=\sum_{\substack{0\leqslant v_i\leqslant u_i\\\sum_i v_i=\ell}}
      \prod_{i=1}^t(-1)^{u_i-v_i}
      q^{\binom{u_i-v_i}{2}}\genfrac{[}{]}{0pt}{}{u_i}{v_i}_q,\\
 \omega_d(\boldsymbol{u})
   &:=\sum_{\ell=d}^{\lvert \boldsymbol{u}\rvert}(q^{m(\ell-d+1)}-1)f_\ell(\boldsymbol{u}),
\end{align*}
where an empty sum is zero. The numbers $\omega_d(\boldsymbol{u})$ are
the formal support counts. For fixed $d$, our
$\omega_d(\boldsymbol{u})$ is $\omega(\boldsymbol{u})$ in
\cite[Equation~(VI.2)]{BGR2021}; the subscript records the distance
parameter. If a nonzero linear MSRD code $C$ of minimum distance $d$
exists and $\boldsymbol{U}\ne(0,\ldots,0)$ is a specified support
tuple with $\dim\boldsymbol{U}=\boldsymbol{u}$, then
\[
 W_{\boldsymbol{U}}(C)=\omega_d(\boldsymbol{u})
   =\lvert\{X\in C:\sigma(X)=\boldsymbol{U}\}\rvert
\]
by \cite[Theorem~VI.5]{BGR2021}. This is the number of codewords whose
support is the \emph{specified} tuple $\boldsymbol{U}$, not the number
of codewords whose support has dimension vector $\boldsymbol{u}$. The
zero tuple must be excluded: exactly one codeword, namely $0$, has zero
support, whereas $\omega_d(\boldsymbol{0})=0$.

For each integer $k$ with $0\leqslant k\leqslant N$, define
$\boldsymbol{g}_k\in\mathcal I$ by
\[
 (\boldsymbol{g}_k)_i:=\min\!\left\{n_i,\,
       \max\!\left\{0,k-\sum_{j<i}n_j\right\}\right\}.
\]
Thus $\boldsymbol{g}_k$ is obtained by filling the coordinates in
order, the $i$-th up to $n_i$, until the total is $k$; in
particular, $\lvert \boldsymbol{g}_k\rvert=k$. We call
$\boldsymbol{g}_k$ the \emph{greedy vector} of total dimension $k$.
The vector $\widetilde{\boldsymbol{u}}:=\boldsymbol{g}_{d+1}$ is the
vector in the conjecture \cite[Conjecture~VI.8]{BGR2021} obtained by filling the
longest blocks first, with the definition extended to the case
$d+1=n_1+\cdots+n_j$. Finally, put
\[
 S(\boldsymbol{u}):=\sum_{i=1}^t\genfrac{[}{]}{0pt}{}{u_i}{1}_q
     =\sum_{i=1}^t\frac{q^{u_i}-1}{q-1}.
\]

\begin{theorem}\label{thm:main}
For every prime power $q$ and all parameters above, the following
conditions are equivalent:
\begin{enumerate}[label=\textup{(\roman*)}]
\item\label{it:all-nonnegative} $\omega_d(\boldsymbol{u})\geqslant0$ for every $\boldsymbol{u}\in \mathcal I$;
\item\label{it:single-test} $S(\widetilde{\boldsymbol{u}})\leqslant q^m+1$.
\end{enumerate}
\end{theorem}

Since $\lvert \widetilde{\boldsymbol{u}}\rvert=d+1$,
Lemma~\ref{lem:coefficient} below gives
$\omega_d(\widetilde{\boldsymbol{u}})=(q^m-1)(q^m+1-S(\widetilde{\boldsymbol{u}}))$.
Condition~\ref{it:single-test} is therefore precisely the inequality
$\omega_d(\widetilde{\boldsymbol{u}})\geqslant0$ that the conjecture
\cite[Conjecture~VI.8]{BGR2021} proposes to check; we refer to it as
the \emph{single test}. Its necessity is immediate, and the content of
Theorem~\ref{thm:main} is the sufficiency.

With $j$ and $\delta$ as in the introduction, that is,
\[
 \sum_{i<j}n_i<d+1\leqslant\sum_{i\leqslant j}n_i,
 \qquad \delta:= d+1-\sum_{i<j}n_i,
\]
we have $1\leqslant\delta\leqslant n_j$ and
$\widetilde{\boldsymbol{u}}=(n_1,\ldots,n_{j-1},\delta,0,\ldots,0)$, so
condition~\ref{it:single-test} is exactly~\eqref{eq:explicit-test}. The
case $d+1=n_1+\cdots+n_j$ corresponds to $\delta=n_j$. For the
parameters $q=m=3$, $(n_1,n_2,n_3)=(3,3,2)$ and $d=7$ of
\cite[Example~VI.7]{BGR2021}, the two sides of~\eqref{eq:explicit-test}
are $60$ and $56$, respectively, in agreement with the negative value
$\omega_7(3,3,2)=-52$ computed there.

\section{A Generating Function and Reduction to Greedy Vectors}
\label{sec:reduction}

This section reduces Theorem~\ref{thm:main} to a statement about the
greedy vectors $\boldsymbol{g}_s$ and rewrites the counts at these
vectors as alternating sums of elementary symmetric polynomials.

\subsection{Notation}

For a finite multiset $X=\{x_1,\ldots,x_L\}$ of real numbers,
its \emph{elementary symmetric polynomials} are
\[
 e_j(X):=\sum_{1\leqslant i_1<\cdots<i_j\leqslant L}
 x_{i_1}\cdots x_{i_j}\qquad(1\leqslant j\leqslant L).
\]
We set $e_0(X):=1$ and $e_j(X):=0$ for $j>L$.
Entries are counted with multiplicity, and adjoining or removing an
entry affects a single occurrence; $\cup$ denotes the union of
multisets, in which multiplicities add.
For a formal power series $F(z)$ over $\mathbb{Q}$, we write
$[z^r]F(z)$ for its coefficient of $z^r$.

\subsection{Formal Support Counts as Coefficients}

For $\boldsymbol{u}\in \mathcal I$ define
\[
 P_{\boldsymbol{u}}(z):=\prod_{i=1}^t\prod_{j=0}^{u_i-1}(1-q^jz).
\]
The finite $q$-binomial theorem \cite[Eq.~(17.2.35)]{DLMF} gives
\[
 \prod_{j=0}^{a-1}(1-q^jz)
 =\sum_{h=0}^{a}(-1)^h q^{\binom h2}
     \genfrac{[}{]}{0pt}{}{a}{h}_qz^h.
\]
Expanding each block factor of $P_{\boldsymbol{u}}(z)$ in this way and using
$\genfrac{[}{]}{0pt}{}{a}{h}_q=\genfrac{[}{]}{0pt}{}{a}{a-h}_q$
shows that, for $s=\lvert \boldsymbol{u}\rvert$,
\[
 [z^{s-\ell}]P_{\boldsymbol{u}}(z)=f_\ell(\boldsymbol{u})\qquad(0\leqslant\ell\leqslant s).
\]

\begin{lemma}\label{lem:coefficient}
Let $\boldsymbol{u}\in \mathcal I$ and $s=\lvert \boldsymbol{u}\rvert$. If $s\geqslant d$, then
\begin{equation}\label{eq:coefficient}
 \frac{\omega_d(\boldsymbol{u})}{q^m-1}
 =[z^{s-d}]\frac{P_{\boldsymbol{u}}(z)}{(1-z)(1-q^mz)}.
\end{equation}
Moreover, for every $\boldsymbol{u}\in\mathcal I$ with $\lvert\boldsymbol{u}\rvert\leqslant d+1$,
\[
 \omega_d(\boldsymbol{u})=\begin{cases}
 0,&\lvert \boldsymbol{u}\rvert<d,\\
 q^m-1,&\lvert \boldsymbol{u}\rvert=d,\\
 (q^m-1)(q^m+1-S(\boldsymbol{u})),&\lvert \boldsymbol{u}\rvert=d+1.
 \end{cases}
\]
\end{lemma}

\begin{proof}
By the identity above, $[z^j]P_{\boldsymbol{u}}(z)=f_{s-j}(\boldsymbol{u})=(-1)^je_j(A)$,
where $A$ is the multiset consisting of $q^h$ for each pair $(i,h)$ with
$1\leqslant i\leqslant t$ and $0\leqslant h<u_i$. Also,
\[
\begin{aligned}
 \relax[z^r]\frac{1}{(1-z)(1-q^mz)}
 &=1+q^m+\cdots+q^{mr}\\
 &=\frac{q^{m(r+1)}-1}{q^m-1}\qquad(r\geqslant0).
\end{aligned}
\]
For $r=s-d\geqslant0$, convolution therefore gives
\[
 \begin{aligned}
 \relax[z^r]\frac{P_{\boldsymbol{u}}(z)}{(1-z)(1-q^mz)}
 &=\frac{1}{q^m-1}\!\sum_{j=0}^r
   (q^{m(r-j+1)}-1)f_{s-j}(\boldsymbol{u})\\
 &=\frac{1}{q^m-1}\sum_{\ell=d}^s
   (q^{m(\ell-d+1)}-1)f_\ell(\boldsymbol{u}),
 \end{aligned}
\]
which is \eqref{eq:coefficient}. For $\lvert\boldsymbol{u}\rvert<d$ the
defining sum of $\omega_d(\boldsymbol{u})$ is empty. For
$\lvert\boldsymbol{u}\rvert=d$ and $d+1$, read off the coefficients of
$z^0$ and $z^1$ in \eqref{eq:coefficient}, using
$e_1(A)=S(\boldsymbol{u})$.
\end{proof}

\subsection{Reduction to Greedy Vectors}

The next lemma compares vectors of the same total dimension by moving
one unit from one coordinate to another.

\begin{lemma}\label{lem:greedy}
Fix $d+1\leqslant s\leqslant N$. If $\omega_d(\boldsymbol{w})\geqslant0$ for every $\boldsymbol{w}\in \mathcal I$
with $\lvert \boldsymbol{w}\rvert=s-1$, then
\[
 \omega_d(\boldsymbol{u})\geqslant\omega_d(\boldsymbol{g}_s)
 \qquad\text{for every }\boldsymbol{u}\in \mathcal I\text{ with }\lvert \boldsymbol{u}\rvert=s.
\]
Consequently, the inequalities $\omega_d(\boldsymbol{u})\geqslant0$ for all
$\boldsymbol{u}\in \mathcal I$ are equivalent to
$\omega_d(\boldsymbol{g}_s)\geqslant0$ for $d+1\leqslant s\leqslant N$.
\end{lemma}

\begin{proof}
Any adjacent inversion $u_i<u_{i+1}$ can be removed by
interchanging the two coordinates. The interchange is feasible because
$n_i\geqslant n_{i+1}$; it leaves $P_{\boldsymbol{u}}$ and hence $\omega_d(\boldsymbol{u})$ unchanged.
After finitely many interchanges we may therefore assume
$u_1\geqslant\cdots\geqslant u_t$.

If $\boldsymbol{u}\ne \boldsymbol{g}_s$, let $i$ be the first coordinate
with $u_i<n_i$. Some $j>i$ has $u_j>0$, since otherwise
$\boldsymbol{u}=\boldsymbol{g}_s$; choose the last such $j$.
Set $\boldsymbol{w}=\boldsymbol{u}-\mathbf e_j$ and $\boldsymbol{u}'=\boldsymbol{u}+\mathbf e_i-\mathbf e_j$. The transfer
$\boldsymbol{u}\mapsto\boldsymbol{u}'$ stays in $\mathcal I$ and keeps the
coordinates nonincreasing.
Indeed, if $i>1$ then $u_{i-1}=n_{i-1}\geqslant n_i\geqslant u_i+1$,
while $u_{j+1}=0\leqslant u_j-1$ if $j<t$; all other
neighboring inequalities are preserved. The product identity
\[
 P_{\boldsymbol{u}'}(z)-P_{\boldsymbol{u}}(z)
 =-(q^{u_i}-q^{u_j-1})zP_{\boldsymbol{w}}(z)
\]
combined with Lemma~\ref{lem:coefficient} yields
\[
 \omega_d(\boldsymbol{u}')-\omega_d(\boldsymbol{u})
 =-(q^{u_i}-q^{u_j-1})\omega_d(\boldsymbol{w})\leqslant0.
\]
The coefficient in parentheses is positive because $i<j$ and $\boldsymbol{u}$
is nonincreasing. Each transfer strictly increases the vector in the
lexicographic order, and $\mathcal I$ is finite, so after finitely many
transfers no positive coordinate lies to the right of the first
coordinate $i$ with $u_i<n_i$. This property characterizes
$\boldsymbol{g}_s$, which proves the inequality.
For the consequence, assume
$\omega_d(\boldsymbol{g}_s)\geqslant0$ for $d+1\leqslant s\leqslant N$.
Lemma~\ref{lem:coefficient} gives nonnegativity through dimension
$d$. Inductively, nonnegativity in dimension $s-1$ and the first
assertion give
$\omega_d(\boldsymbol{u})\geqslant\omega_d(\boldsymbol{g}_s)\geqslant0$
for all $\boldsymbol{u}$ of total dimension $s$.
The converse follows by taking $\boldsymbol{u}=\boldsymbol{g}_s$.
\end{proof}

\subsection{Normalization}

By Lemma~\ref{lem:greedy}, it remains to study the counts at the
vectors $\boldsymbol{g}_s$. Let
$\alpha:=(\alpha_1,\ldots,\alpha_N)$ be the sequence obtained by concatenating
$(1,q,\ldots,q^{n_i-1})$ for $i=1,\ldots,t$, in that order.
By construction, $P_{\boldsymbol{g}_s}(z)=\prod_{j=1}^s(1-\alpha_jz)$.
Put
\[
\begin{aligned}
 B&:=\{\alpha_2/q^m,\ldots,\alpha_{d+1}/q^m\},\\
 X_s&:=\{\alpha_2/q^m,\ldots,\alpha_s/q^m\}\quad(s\geqslant d+1).
\end{aligned}
\]
Thus $X_{d+1}=B$. We call the numbers $\alpha_j/q^m$ the
\emph{normalized entries}; each of them lies in $(0,1/q]$. Since $\alpha_1=1$ and
$\alpha_1+\cdots+\alpha_{d+1}=S(\boldsymbol{g}_{d+1})$, the single test
$S(\boldsymbol{g}_{d+1})\leqslant q^m+1$ is equivalent to
\begin{equation}\label{eq:initial-sum}
 \sum_{x\in B}x=\frac{S(\boldsymbol{g}_{d+1})-1}{q^m}\leqslant1.
\end{equation}
For a finite multiset $X$ and an integer $r\geqslant0$, put
\[
 C_r(X):=\sum_{j=0}^r(-1)^j e_j(X).
\]
The polynomial $P_{\boldsymbol{g}_s}$ contains the factor
$1-z$ arising from $\alpha_1=1$. Canceling it in
\eqref{eq:coefficient} and expanding $1/(1-q^mz)$ gives, for
$r=s-d\geqslant1$,
\begin{equation}\label{eq:normalized}
 \omega_d(\boldsymbol{g}_s)=(q^m-1)q^{mr} C_r(X_s).
\end{equation}

\section{Alternating Sums of Elementary Symmetric Polynomials}
\label{sec:positivity}

By Lemma~\ref{lem:greedy} and~\eqref{eq:normalized}, the implication
\ref{it:single-test}$\Rightarrow$\ref{it:all-nonnegative} of
Theorem~\ref{thm:main} amounts to showing that~\ref{it:single-test}
implies $C_r(X_s)\geqslant0$ for $d+1\leqslant s\leqslant N$. We
carry this out for $q\geqslant3$ and also record the estimates needed
for $q=2$. Recall that $e_1(X)=\sum_{x\in X}x$.

\subsection{Elementary Estimates}

\begin{lemma}\label{lem:symmetric}
Let $X$ be a finite multiset of real numbers in $[0,1]$. For every
nonnegative integer $h$ and positive integer $r$,
\begin{align*}
 C_{2h}(X)&\geqslant\prod_{x\in X}(1-x)\geqslant0,\\
 r e_r(X)&\leqslant e_1(X)e_{r-1}(X).
\end{align*}
Moreover, for every finite multiset $W$ in $[0,1]$, every odd
$r\geqslant3$ and $x,y\in[0,1]$,
\begin{equation}\label{eq:two-factor}
 \begin{split}
 C_r(W\cup\{x,y\})={}&(1-x)(1-y)C_{r-2}(W)\\
 &+(1-x-y)e_{r-1}(W)-e_r(W).
 \end{split}
\end{equation}
\end{lemma}

\begin{proof}
For the first inequality, induct on the size of $X$. Adjoining an
entry $x\in[0,1]$ gives
\[
 C_{2h}(X\cup\{x\})
 =(1-x)C_{2h}(X)+xe_{2h}(X),
\]
which preserves the stated lower bound. For the second, index the
entries of $X$ as $x_1,\ldots,x_L$ and expand
$e_1(X)e_{r-1}(X)=\sum_i\sum_J x_i\prod_{j\in J}x_j$, where $J$ runs
over the $(r-1)$-element subsets of $\{1,\ldots,L\}$. The terms with
$i\notin J$ contribute $r e_r(X)$, since each $r$-element set $K$
arises from the $r$ choices of $i\in K$. Hence
\begin{equation}\label{eq:er-gap}
 e_1(X)e_{r-1}(X)-r e_r(X)
 =\sum_{\lvert J\rvert=r-1}\Bigl(\prod_{j\in J}x_j\Bigr)
   \Bigl(\sum_{j\in J}x_j\Bigr)\geqslant0.
\end{equation}
Expanding
$C_r(W\cup\{x,y\})$ and using
the identities
\[
\begin{aligned}
C_{r-1}(W)&=C_{r-2}(W)+e_{r-1}(W),\\
C_r(W)&=C_{r-2}(W)+e_{r-1}(W)-e_r(W)
\end{aligned}
\]
proves \eqref{eq:two-factor}.
\end{proof}

\begin{lemma}\label{lem:one-third}
Let $B$ be a finite multiset in $[0,1]$ such that
$e_1(B)\leqslant1$, and let $y_1,\ldots,y_h\in[0,1/3]$, where $h\geqslant0$.
Then
\[
 C_{h+1}(B\cup\{y_1,\ldots,y_h\})\geqslant0.
\]
\end{lemma}

\begin{proof}
If $h+1$ is even, the claim follows from Lemma~\ref{lem:symmetric}.
For odd $r=h+1$, proceed by induction in steps of two. The case $r=1$
is $C_1(B)=1-e_1(B)\geqslant0$. For $r\geqslant3$, set
\[
 W=B\cup\{y_1,\ldots,y_{r-3}\},
 \qquad x=y_{r-2},\quad y=y_{r-1}.
\]
The induction hypothesis applies to $W$, which has $r-3$ additional
entries, and gives $C_{r-2}(W)\geqslant0$. Moreover,
$e_1(W)\leqslant1+(r-3)/3=r/3$, so Lemma~\ref{lem:symmetric}
gives $e_r(W)\leqslant e_{r-1}(W)/3$. Since $x+y\leqslant2/3$,
\eqref{eq:two-factor} yields
\[
 C_r(W\cup\{x,y\})
 \geqslant(1-x)(1-y)C_{r-2}(W)\geqslant0,
\]
which completes the induction.
\end{proof}

\begin{lemma}\label{lem:complement}
Let $W$ be a multiset of $L$ real numbers in $[0,a]$.
For $1\leqslant r\leqslant L+1$,
\begin{equation}\label{eq:complement}
 r e_r(W)\leqslant(L-r+1)a\,e_{r-1}(W).
\end{equation}
\end{lemma}

\begin{proof}
Index the entries of $W$ as $w_1,\ldots,w_L$. As in the proof of
Lemma~\ref{lem:symmetric}, $r e_r(W)$ is the sum, over all
$(r-1)$-element sets $I$ of indices, of $\prod_{i\in I}w_i$ times the
sum of the $L-r+1$ entries $w_j$ with $j\notin I$. Each of these
sums is at most $(L-r+1)a$, which proves \eqref{eq:complement}; for
$r=L+1$ both sides vanish.
\end{proof}

\subsection{Fields With at Least Three Elements}

\begin{proposition}\label{prop:q-at-least-3}
If $q\geqslant3$, then condition~\ref{it:single-test} of
Theorem~\ref{thm:main} implies condition~\ref{it:all-nonnegative}.
\end{proposition}

\begin{proof}
For $q\geqslant3$, the normalized entries $\alpha_j/q^m$ with $j>d+1$
all belong to $(0,1/q]\subseteq[0,1/3]$, and
condition~\ref{it:single-test} is~\eqref{eq:initial-sum}. Since
$X_s=B\cup\{\alpha_{d+2}/q^m,\ldots,\alpha_s/q^m\}$ consists of $B$
and $s-d-1=r-1$ further entries, where $r=s-d$,
Lemma~\ref{lem:one-third} gives
$C_r(X_s)\geqslant0$ for $d+1\leqslant s\leqslant N$. The assertion
follows from~\eqref{eq:normalized} and Lemma~\ref{lem:greedy}.
\end{proof}

\subsection{Small Distances}

The two smallest distances require no hypothesis on
$S(\boldsymbol{g}_{d+1})$, for any $q$.

\begin{lemma}\label{lem:small-distance}
If $d=1$ or $d=2$, then $\omega_d(\boldsymbol{u})\geqslant0$ for every
$\boldsymbol{u}\in\mathcal I$.
\end{lemma}

\begin{proof}
Since $\omega_d(\boldsymbol{u})=0$ when $\lvert\boldsymbol{u}\rvert<d$,
we may assume $s=\lvert \boldsymbol{u}\rvert\geqslant d$. Then
$s\geqslant1$, so $P_{\boldsymbol{u}}(z)=(1-z)\prod_{a\in A}(1-az)$, where $A$ is a
multiset of $s-1$ numbers $q^j$ with $0\leqslant j<m$.
Put $r=s-d$ and $X=\{a/q^m:a\in A\}$. Canceling $1-z$ in
\eqref{eq:coefficient} gives
\[
 \omega_d(\boldsymbol{u})=(q^m-1)q^{mr} C_r(X),\qquad r\geqslant0.
\]
Every entry of $X$ belongs to $(0,1/q]\subseteq(0,1/2]$. For $d=1$, the multiset $X$ has $r$ entries and
$C_r(X)=\prod_{x\in X}(1-x)>0$.
For $d=2$, $X$ has $r+1$ entries; omitting the top-degree term gives
\[
 C_r(X)=\prod_{x\in X}(1-x)
            -(-1)^{r+1}\prod_{x\in X}x.
\]
For even $r$ this is positive, and for odd $r$ it is nonnegative
because $1-x\geqslant x$ for each entry. In both cases the displayed
formula for $\omega_d(\boldsymbol{u})$ gives the assertion.
\end{proof}

\section{The Binary Case and the Proof of the Main Theorem}
\label{sec:binary}

Throughout this section $q=2$. A block of length $a$ contributes the
normalized entries
\[
 \frac{1}{q^m},\frac{2}{q^m},\ldots,\frac{2^{a-1}}{q^m}.
\]
The only entries greater than $1/3$ are those equal to $1/2$;
they occur exactly at the end of a block of length $m$. The next
lemma restricts where such an entry can occur under the single test.

\begin{lemma}\label{lem:half-entry}
Suppose $m\geqslant2$, $S(\boldsymbol{g}_{d+1})\leqslant q^m+1$, and
$\alpha_j=q^m/2$ for some $j>d+1$. Then $n_1=m$ and
$d\leqslant m$.
\end{lemma}

\begin{proof}
An entry $q^m/2$ is the last entry of a block of length $m$.
Since the block lengths are nonincreasing, all preceding blocks also
have length $m$; in particular $n_1=m$. If this entry lies in the first block, its
position is $m\geqslant d+2$, hence $d\leqslant m-2$. Otherwise the first two
blocks have length $m$. If $d\geqslant m+1$, the first $d+1$ entries
contain all $m$ entries of the first block and the first two entries
of the second block, so their sum is at least
\[
 (q^m-1)+1+2=q^m+2,
\]
contrary to $S(\boldsymbol{g}_{d+1})\leqslant q^m+1$.
\end{proof}

For $m\geqslant4$, the normalized entries satisfy a sharper estimate.

\begin{lemma}\label{lem:binary-odd}
Assume $m\geqslant4$ and
$S(\boldsymbol{g}_{d+1})\leqslant q^m+1$. For odd
$r=s-d\geqslant3$ with $s\leqslant N$, let
\[
 W=\{\alpha_2/q^m,\ldots,\alpha_{s-2}/q^m\},
 \qquad x=\frac{\alpha_{s-1}}{q^m},\quad
 y=\frac{\alpha_s}{q^m}.
\]
Then
\[
 r(1-x-y)-e_1(W)\geqslant\frac{3}{q^m}.
\]
The constant $3$ is best possible.
\end{lemma}

\begin{proof}
Put $\Delta=r(1-x-y)-e_1(W)$. The multiset
$W$ consists of $B$ and $r-3$ further normalized entries, each at
most $1/2$. By~\eqref{eq:initial-sum},
$e_1(W)\leqslant(r-1)/2$. Thus if $x+y\leqslant1/2$, then
$\Delta\geqslant1/2\geqslant3/q^m$.

Suppose $x+y>1/2$. As $m\geqslant4$, adjacent normalized entries
have sum greater than $1/2$ only for
$(x,y)=(1/4,1/2)$, the final two entries of a length-$m$ block,
or $(x,y)=(1/2,1/q^m)$, the pair crossing the boundary after such a
block. In either case a normalized entry $1/2$ occurs beyond the
first $d+1$ positions. By Lemma~\ref{lem:half-entry}, $d\leqslant m$.

\textit{Case $(x,y)=(1/4,1/2)$.} The first $s$ entries comprise
$\ell$ blocks of length $m$ for some $\ell\geqslant1$. Hence
$s=\ell m$, $r=\ell m-d$, and
\[
 \Delta=\frac{\ell(m-4)+3-d}{4}+\frac{\ell+1}{q^m}.
\]
If $m=4$, the oddness of $r$ gives $d\in\{1,3\}$.
For $d=3$ the condition $r\geqslant3$ requires $\ell\geqslant2$;
for $d=1$ the first term is $1/2$. Thus
$\Delta\geqslant3/q^m$. If $m\geqslant5$ and $\ell\geqslant2$,
then
\[
 \Delta\geqslant\frac{m-5}{4}+\frac{\ell+1}{q^m}
 \geqslant\frac{3}{q^m}.
\]
If instead $m\geqslant5$ and $\ell=1$, then $r\geqslant3$ gives
$\Delta=(r-1)/4+2/q^m\geqslant3/q^m$.

\textit{Case $(x,y)=(1/2,1/q^m)$.} Now $s=\ell m+1$ and
$r=\ell m+1-d$, where the first $\ell$ blocks have length $m$.
A direct calculation gives
\[
\begin{aligned}
 \Delta={}&(m-d)\bigl(\tfrac12-\tfrac{1}{q^m}\bigr)+\tfrac{1}{q^m}\\
 &+(\ell-1)\bigl(\tfrac{m-2}{2}-\tfrac{m-1}{q^m}\bigr).
\end{aligned}
\]
For $\ell=1$, the restriction $r\geqslant3$ implies
$m-d\geqslant2$, and therefore $\Delta\geqslant1-1/q^m\geqslant3/q^m$.
For $\ell\geqslant2$, use $d\leqslant m$ and
$(m-2)/2-(m-1)/q^m\geqslant2/q^m$ (valid for $m\geqslant4$)
to obtain $\Delta\geqslant3/q^m$.

Finally, for $m=4$, two blocks of length $4$, $d=3$, and $s=8$,
one has $r=5$, $x=1/4$, $y=1/2$, and $e_1(W)=17/16$.
Thus $\Delta=3/16$, proving sharpness.
\end{proof}

\begin{proposition}\label{prop:binary}
If $q=2$, then condition~\ref{it:single-test} of
Theorem~\ref{thm:main} implies condition~\ref{it:all-nonnegative}.
\end{proposition}

\begin{proof}
We first reduce to the case in which $q^m/2$ occurs among
$\alpha_{d+2},\ldots,\alpha_N$, and then prove $C_r(X_s)\geqslant0$ by
induction on odd $r$.

The cases $d\leqslant2$ follow from Lemma~\ref{lem:small-distance}.
Assume $d\geqslant3$. If $m=1$, then $S(\boldsymbol{g}_{d+1})=d+1$, and
the hypothesis gives $d+1\leqslant q^m+1=3$, a contradiction; hence
$m\geqslant2$.

If $\alpha_j\ne q^m/2$ for all $j>d+1$, then $\alpha_j/q^m\leqslant1/4$
for all $j>d+1$. In this case
Lemma~\ref{lem:one-third}, together with \eqref{eq:normalized}
and Lemma~\ref{lem:greedy}, proves the assertion.
We may therefore assume that $\alpha_j=q^m/2$ for some $j>d+1$.
Lemma~\ref{lem:half-entry} then gives $d\leqslant m$,
so $m\geqslant3$.

For $s\geqslant d+1$ write $r=s-d$ and use the multiset $X_s$
from Section~\ref{sec:reduction}. If $r$ is even, then
$C_r(X_s)>0$ by Lemma~\ref{lem:symmetric}, because every entry of
$X_s$ is less than $1$. The odd case is
proved by induction on $r$. At $r=1$,
$C_1(X_{d+1})=1-e_1(B)\geqslant0$ by \eqref{eq:initial-sum}.
For odd $r\geqslant3$, set
\[
 W=X_{s-2},\qquad x=\frac{\alpha_{s-1}}{q^m},
 \qquad y=\frac{\alpha_s}{q^m}.
\]
By \eqref{eq:two-factor},
\[
 \begin{split}
 C_r(X_s)={}&(1-x)(1-y)C_{r-2}(W)\\
 &+(1-x-y)e_{r-1}(W)-e_r(W).
 \end{split}
\]
Since $W=X_{s-2}$, the first term is nonnegative by the induction
hypothesis. It remains to show that
$e_r(W)\leqslant(1-x-y)e_{r-1}(W)$.

\textit{Case $m\geqslant4$.} By Lemmas~\ref{lem:symmetric}
and~\ref{lem:binary-odd},
$r e_r(W)\leqslant e_1(W)e_{r-1}(W)\leqslant r(1-x-y)e_{r-1}(W)$.

\textit{Case $m=3$.} Here $d\geqslant3$ and $d\leqslant m$
force $d=3$. Within each block, the normalized entries are an initial segment
of $(1/8,1/4,1/2)$; across a boundary, the new entry is $1/8$.
Every adjacent pair therefore satisfies $x+y\leqslant3/4$,
and $\lvert W\rvert=r+d-3=r$. Since the entries of $W$ are at most
$1/2$, Lemma~\ref{lem:complement} gives
\[
 e_r(W)\leqslant\frac{1}{2r}e_{r-1}(W)
 \leqslant\frac16 e_{r-1}(W)
 \leqslant(1-x-y)e_{r-1}(W).
\]
This completes the induction, so $C_r(X_s)\geqslant0$ for all odd $r$.
By \eqref{eq:normalized}, $\omega_d(\boldsymbol{g}_s)\geqslant0$ for
$d+1\leqslant s\leqslant N$, and Lemma~\ref{lem:greedy} gives the
assertion for every $\boldsymbol{u}\in\mathcal I$.
\end{proof}

\begin{proof}[Proof of Theorem~\ref{thm:main}]
The implication
\ref{it:single-test}$\Rightarrow$\ref{it:all-nonnegative} is
Proposition~\ref{prop:q-at-least-3} for $q\geqslant3$ and
Proposition~\ref{prop:binary} for $q=2$. Conversely, assume~\ref{it:all-nonnegative}.
Applying it to $\boldsymbol{g}_{d+1}$ and using Lemma~\ref{lem:coefficient} gives
\[
 0\leqslant\omega_d(\boldsymbol{g}_{d+1})=(q^m-1)(q^m+1-S(\boldsymbol{g}_{d+1})).
\]
Since $q^m>1$, this is condition~\ref{it:single-test}.
\end{proof}

\section{The Dual Conditions and Comparison With Sphere Packing}
\label{sec:dual}

\subsection{The Dual Formal Support Counts}

The dual necessary conditions \cite[Corollary~VI.10]{BGR2021}
concern the numbers $\widehat\omega(\boldsymbol{u})$ for
$\boldsymbol{u}\in \mathcal I$, which we write as
$\widehat\omega_d(\boldsymbol{u})$:
\[
 \widehat\omega_d(\boldsymbol{u})
 :=\sum_{\ell=N-d+2}^{\lvert \boldsymbol{u}\rvert}
       (q^{m(\ell-N+d-1)}-1)f_\ell(\boldsymbol{u}).
\]
Comparison with the definition of $\omega_d$ shows that
$\widehat\omega_d=\omega_{N-d+2}$. For $d\leqslant2$ no dual test is
needed: $\widehat\omega_1$ vanishes identically, and
$\widehat\omega_2(\boldsymbol{u})$ vanishes unless
$\lvert \boldsymbol{u}\rvert=N$, in which case it equals $q^m-1$. For
$d\geqslant3$ we have $N-d+2<N$, so Theorem~\ref{thm:main} applies to
$\omega_{N-d+2}$.

\begin{corollary}\label{cor:dual}
Suppose $3\leqslant d<N$. Then
\[
 \widehat\omega_d(\boldsymbol{u})\geqslant0\text{ for all }\boldsymbol{u}\in \mathcal I
 \quad\Longleftrightarrow\quad
 S(\boldsymbol{g}_{N-d+3})\leqslant q^m+1.
\]
\end{corollary}

\begin{proof}
Put $D=N-d+2$. The hypotheses give $1\leqslant D<N$, and
$\widehat\omega_d(\boldsymbol{u})=\omega_D(\boldsymbol{u})$ identically for all $\boldsymbol{u}$.
Apply Theorem~\ref{thm:main} with distance parameter $D$.
\end{proof}

\begin{corollary}\label{cor:joint}
Suppose $3\leqslant d<N$ and put $K:=\max\{d+1,N-d+3\}$.
The numbers $\omega_d(\boldsymbol{u})$ and
$\widehat\omega_d(\boldsymbol{u})$ are both nonnegative for every
$\boldsymbol{u}\in\mathcal I$ if and only if
\[
 S(\boldsymbol{g}_K)\leqslant q^m+1.
\]
In particular, this inequality is necessary for the existence of an
$\mathbb F_q$-linear MSRD code with these parameters.
\end{corollary}

\begin{proof}
By Theorem~\ref{thm:main} and Corollary~\ref{cor:dual}, the two
families are equivalent to $S(\boldsymbol{g}_{d+1})\leqslant q^m+1$ and
$S(\boldsymbol{g}_{N-d+3})\leqslant q^m+1$, respectively.
For $0\leqslant k<N$, the vector $\boldsymbol{g}_{k+1}$ is obtained from
$\boldsymbol{g}_k$ by increasing one coordinate from some value $a$ to
$a+1$, so
\[
 S(\boldsymbol{g}_{k+1})-S(\boldsymbol{g}_k)=q^a>0.
\]
Thus $S(\boldsymbol{g}_k)$ is strictly increasing in $k$, and the two
inequalities together are equivalent to the one with the larger index,
namely $S(\boldsymbol{g}_K)\leqslant q^m+1$. The last assertion follows
from the necessary support-distribution conditions
\cite[Corollaries~VI.6 and~VI.10]{BGR2021}.
\end{proof}

\subsection{Comparison With Sphere Packing}

The sphere-packing argument below retains a prescribed number of rows
across the blocks. To optimize this bound, we first determine
which distribution of these rows maximizes $S$.

\begin{lemma}\label{lem:S-greedy}
Let $\boldsymbol{v}\in\mathcal I$ and $k=\lvert\boldsymbol{v}\rvert$.
Then $S(\boldsymbol{v})\leqslant S(\boldsymbol{g}_k)$, with equality if
and only if $\boldsymbol{v}$ is obtained from $\boldsymbol{g}_k$ by
permuting coordinates.
\end{lemma}

\begin{proof}
Apply to $\boldsymbol{v}$ the interchanges and transfers in the proof
of Lemma~\ref{lem:greedy}. They use only
$n_1\geqslant\cdots\geqslant n_t$, not the hypothesis on $\omega_d$,
and they lead to $\boldsymbol{g}_k$. An interchange leaves $S$
unchanged, and a transfer
$\boldsymbol{u}\mapsto\boldsymbol{u}+\mathbf e_i-\mathbf e_j$
increases $S$ by $q^{u_i}-q^{u_j-1}$, which is positive because
$u_i\geqslant u_j$. At least one transfer occurs unless the sorted
vector is $\boldsymbol{g}_k$, that is, unless $\boldsymbol{v}$ is a
permutation of $\boldsymbol{g}_k$.
\end{proof}

For a code $C\subseteq\Pi$, let
$C^\perp:=\{Y\in\Pi:\langle X,Y\rangle=0\text{ for every }X\in C\}$
be its dual with respect to the trace-product pairing
$\langle X,Y\rangle:=\sum_{i=1}^{t}\operatorname{Tr}(X_iY_i^{\mathsf T})$
\cite[Definition~V.1]{BGR2021}, where $\operatorname{Tr}$ is the
ordinary matrix trace over $\mathbb F_q$.

\begin{proposition}\label{prop:sphere-comparison}
Let $A\subseteq\Pi$ be a code of minimum distance $d_A$, where
$3\leqslant d_A<N$, and put $L:=N-d_A+3$. Then
\[
 \lvert A\rvert\bigl(1+(q^m-1)S(\boldsymbol{v})\bigr)\leqslant q^{mL}
\]
for every $\boldsymbol{v}\in\mathcal I$ with
$\lvert\boldsymbol{v}\rvert=L$, and the choice
$\boldsymbol{v}=\boldsymbol{g}_L$ gives the strongest of these
inequalities.

If $3\leqslant d<N$ and $C\subseteq\Pi$ is an $\mathbb F_q$-linear
MSRD code of minimum distance $d$, then the inequality above for $A=C$
and $\boldsymbol{v}=\boldsymbol{g}_{N-d+3}$ is equivalent to
$S(\boldsymbol{g}_{N-d+3})\leqslant q^m+1$, and the one for
$A=C^\perp$ and $\boldsymbol{v}=\boldsymbol{g}_{d+1}$ is equivalent to
$S(\boldsymbol{g}_{d+1})\leqslant q^m+1$. Together these are
equivalent to the necessary inequality of Corollary~\ref{cor:joint}.
\end{proposition}

\begin{proof}
Let $\boldsymbol{v}\in\mathcal I$ with $\lvert\boldsymbol{v}\rvert=L$.
Delete $n_i-v_i$ rows of block $i$ for every $i$, that is, $d_A-3$
rows in total. Deleting one row lowers matrix rank by at most one,
so puncturing lowers the sum-rank distance between any two codewords
by at most $d_A-3$. Since $d_A-3<d_A$, puncturing is injective on $A$,
and the punctured code has $\lvert A\rvert$ codewords, minimum distance
at least $3$, and an ambient space with $q^{mL}$ elements.

Apply the projective radius-one sphere-packing argument
\cite[Theorem~III.7]{BGR2021} to this punctured code. The radius-one
balls around its codewords are disjoint. Such a ball consists of its
center and the tuples that differ from the center by a rank-one matrix
in a single block; since $\mathbb F_q^{v_i\times m}$ contains
$(q^m-1)\genfrac{[}{]}{0pt}{}{v_i}{1}_q$ matrices of rank one, the ball
has $1+(q^m-1)S(\boldsymbol{v})$ elements. This gives the displayed
inequality, and since $S(\boldsymbol{v})\leqslant S(\boldsymbol{g}_L)$
by Lemma~\ref{lem:S-greedy}, the choice $\boldsymbol{v}=\boldsymbol{g}_L$
gives the strongest one.

If $A$ is MSRD, then $\lvert A\rvert=q^{m(N-d_A+1)}=q^{m(L-2)}$, and the
inequality for $\boldsymbol{v}=\boldsymbol{g}_L$ is equivalent to
$S(\boldsymbol{g}_L)\leqslant q^m+1$. For $A=C$ we have $L=N-d+3$. By
the MSRD duality theorem \cite[Theorem~VI.1]{BGR2021}, $C^\perp$ is
MSRD of minimum distance $D=N-d+2$, and $3\leqslant D<N$ because
$3\leqslant d<N$; for $A=C^\perp$ we therefore have $L=N-D+3=d+1$.
Finally, $S(\boldsymbol{g}_k)$ is strictly increasing in $k$ (see the
proof of Corollary~\ref{cor:joint}), so the two inequalities together
are equivalent to $S(\boldsymbol{g}_K)\leqslant q^m+1$.
\end{proof}

\subsection{The Support-Distribution Conditions Do Not Imply Existence}

The inequality of Corollary~\ref{cor:joint} is necessary for the
existence of a linear MSRD code, but it is not sufficient.

\begin{example}\label{ex:nonexistence}
Let $q\geqslant5$ be an odd prime power, and take $m=1$, $t=N=q+2$,
$n_1=\cdots=n_t=1$, and $d=q$, so that $N-d+2=4$. Since
$S(\boldsymbol{g}_k)=k$, the inequalities of Theorem~\ref{thm:main} and
Corollary~\ref{cor:dual} both hold:
\[
 S(\boldsymbol{g}_{d+1})=q+1=q^m+1,\qquad S(\boldsymbol{g}_{N-d+3})=5\leqslant q^m+1.
\]
Nevertheless, no $\mathbb F_q$-linear MSRD code has these parameters.

\textit{Nonexistence.} For $m=1$ the sum-rank weight is the Hamming
weight, so such a code would be a linear MDS code of length $q+2$,
dimension $3$, and minimum distance $q$. Any three columns of a
generator matrix $G\in\mathbb F_q^{3\times(q+2)}$ are linearly
independent, for otherwise some nonzero codeword would vanish in three
coordinates and have weight at most $q-1$. The columns therefore span
$q+2$ distinct one-dimensional subspaces $A_1,\ldots,A_{q+2}$ of
$\mathbb F_q^3$, no three of which lie in a common plane, that is, a
two-dimensional subspace.

Fix $A_i$. There are exactly $q+1$ planes containing $A_i$, and the
$q+1$ other $A_j$ lie in distinct such planes. Hence every plane
containing one of the $A_i$ contains exactly two of them.

Since $\mathbb F_q^3$ has $q^2+q+1>q+2$ one-dimensional subspaces,
we may choose one, $R$, different from all the $A_i$. Each $A_i$ lies
in exactly one of the $q+1$ planes through $R$, namely $R+A_i$.
Each of these planes contains either zero or two of the $A_i$, so the
$q+2$ subspaces are partitioned into pairs. Thus $q+2$ is even,
contrary to $q$ being odd.

\textit{Direct check of nonnegativity.} The nonnegativity of all
$\omega_q(\boldsymbol{u})$ and $\widehat\omega_q(\boldsymbol{u})$ can
also be verified without Theorem~\ref{thm:main} or
Corollary~\ref{cor:dual}. For $s=\lvert \boldsymbol{u}\rvert$ the
definition gives $f_\ell(\boldsymbol{u})=(-1)^{s-\ell}\binom{s}{\ell}$,
and hence
\[
 \omega_q(\boldsymbol{u})=\begin{cases}
 0,&s<q,\\
 q-1,&s=q,\\
 0,&s=q+1,\\
 q(q-1)^2/2,&s=q+2.
 \end{cases}
\]
For $s<4$ the sum defining $\widehat\omega_q(\boldsymbol{u})$ is empty.
For $s\geqslant4$ put $r=s-4$. Since $P_{\boldsymbol{u}}(z)=(1-z)^s$,
Lemma~\ref{lem:coefficient} with $d=4$ gives
\[
 \widehat\omega_q(\boldsymbol{u})=(q-1)q^r
   \sum_{j=0}^r(-1)^j\binom{s-1}{j}q^{-j}.
\]
Write $a_j=\binom{s-1}{j}q^{-j}$. If $s\leqslant q+1$, then
$a_0\geqslant a_1\geqslant\cdots\geqslant a_r\geqslant0$, so the
alternating sum is nonnegative. If $s=q+2$, then $r=q-2$ is odd and at
least $3$, and $a_1>a_0$ but $a_{j+1}\leqslant a_j$ for $j\geqslant1$;
separating the first four terms gives
\[
 \sum_{j=0}^r(-1)^j a_j
 =\frac{(q-1)(2q-1)}{6q^2}
  +\sum_{k=2}^{(r-1)/2}(a_{2k}-a_{2k+1})>0.
\]

\textit{The endpoint $d=4$.} For the same block parameters and
$3\leqslant d<N$, Corollary~\ref{cor:joint} shows that all
$\omega_d(\boldsymbol{u})$ and $\widehat\omega_d(\boldsymbol{u})$ are
nonnegative exactly when $\max\{d+1,q+5-d\}\leqslant q+1$,
that is, when $4\leqslant d\leqslant q$. At $d=4$, where
$\omega_4=\widehat\omega_q$ and $\widehat\omega_4=\omega_q$, there is
again no linear MSRD code: a linear $[q+2,q-1,4]_q$ MDS code would
have a $3\times(q+2)$ parity-check matrix with every three columns
linearly independent, which the argument above rules out. For $4<d<q$
these numbers are also all nonnegative, but we make no claim about the
existence of codes.
\end{example}

\appendix
\renewcommand{\theequation}{\Alph{section}.\arabic{equation}}
\section{Further Properties of Formal Support Counts}
\label{app:refinements}

The parameters are as in Section~\ref{sec:setup}. Most of the results
below assume the single test; Lemma~\ref{lem:integrality} and
Proposition~\ref{prop:congruences} need no hypothesis on
$S(\boldsymbol{g}_{d+1})$.

\subsection{Strict Positivity}

\begin{corollary}\label{cor:strict}
Suppose $S(\boldsymbol{g}_{d+1})\leqslant q^m+1$. Unless $(q,m,d)=(2,1,2)$,
\[
 \omega_d(\boldsymbol{u})>0\quad\text{whenever }\boldsymbol{u}\in \mathcal I
 \text{ and }\lvert \boldsymbol{u}\rvert\geqslant d+2.
\]
In the exceptional case, all block lengths are $1$, and for
$s=\lvert \boldsymbol{u}\rvert\geqslant2$ one has
\[
 \omega_2(\boldsymbol{u})=\begin{cases}
 1,&s\text{ even},\\
 0,&s\text{ odd}.
 \end{cases}
\]
In every case, the values at $\lvert \boldsymbol{u}\rvert<d$ vanish, those at
$\lvert \boldsymbol{u}\rvert=d$ equal $q^m-1$, and those at $\lvert \boldsymbol{u}\rvert=d+1$
vanish precisely when $S(\boldsymbol{u})=q^m+1$.
\end{corollary}

\begin{proof}
By Theorem~\ref{thm:main} and Lemma~\ref{lem:greedy}, it suffices
to show strict positivity at $\boldsymbol{g}_s$ for $s\geqslant d+2$,
that is, by~\eqref{eq:normalized}, $C_r(X_s)>0$ for $r=s-d\geqslant2$; the statements about
$\lvert\boldsymbol{u}\rvert\leqslant d+1$ follow from
Lemma~\ref{lem:coefficient}.

\textit{Even $r$.} By Lemma~\ref{lem:symmetric},
$C_r(X_s)\geqslant\prod_{x\in X_s}(1-x)>0$.

\textit{The cases $d=1,2$.} These can be read directly
from the proof of Lemma~\ref{lem:small-distance}. For $d=1$,
$C_r(X_s)$ is a strictly positive product. For $d=2$
and odd $r$, it equals
\[
 \prod_{x\in X_s}(1-x)-\prod_{x\in X_s}x.
\]
This is positive unless every entry equals $1/2$.
If $q\geqslant3$, all entries are smaller than $1/2$.
If $q=2$ and $m\geqslant2$, then among $\alpha_2,\ldots,\alpha_s$, where
$s\geqslant4$, some term is smaller than $q^m/2$: if $n_1\geqslant3$,
then $\alpha_2=2<q^m/2$; otherwise $\alpha_{n_1+1}=1$ is the first term
of the second block.
In the remaining case $(q,m,d)=(2,1,2)$, every block has length $1$.
Equation~\eqref{eq:coefficient} gives, for $s\geqslant2$,
\[
 \omega_2(\boldsymbol{u})=2^{s-2}\sum_{j=0}^{s-2}
 (-1)^j\binom{s-1}{j}2^{-j}
 =\frac{1-(-1)^{s-1}}2.
\]

\textit{Odd $r\geqslant3$ and $d\geqslant3$.} Put
$W=X_{s-2}$, $x=\alpha_{s-1}/q^m$, and $y=\alpha_s/q^m$.
The multiset $W$ has $d+r-3\geqslant r-1$ positive entries, so the
right side of~\eqref{eq:er-gap} is positive, that is,
$r e_r(W)<e_1(W)e_{r-1}(W)$.
If $q\geqslant3$, or if $q=2$ and $\alpha_j\ne q^m/2$ for all $j>d+1$,
then $\alpha_j/q^m\leqslant1/3$ for all $j>d+1$. Hence $e_1(W)\leqslant r/3\leqslant r(1-x-y)$.
The strict inequality above and \eqref{eq:two-factor}, with
$C_{r-2}(W)\geqslant0$, imply $C_r(X_s)>0$.

Otherwise $q=2$ and $\alpha_j=q^m/2$ for some $j>d+1$. By
Lemma~\ref{lem:half-entry}, $d\leqslant m$. If $m\geqslant4$,
Lemma~\ref{lem:binary-odd} gives
$e_1(W)\leqslant r(1-x-y)$, and the same strict inequality
implies $C_r(X_s)>0$. If $m=3$, then $d=3$ and
Lemma~\ref{lem:complement} gives
\[
 e_r(W)\leqslant\frac{e_{r-1}(W)}{2r}
 < (1-x-y)e_{r-1}(W),
\]
because $r\geqslant3$ and $x+y\leqslant3/4$.
Thus \eqref{eq:two-factor} again yields $C_r(X_s)>0$.
\end{proof}

\subsection{Lower Bounds in Total Dimension \texorpdfstring{$d+r$}{d+r} for Odd \texorpdfstring{$r$}{r}}

\begin{proposition}\label{prop:binary-odd-refinement}
Suppose $d\geqslant3$, $S(\boldsymbol{g}_{d+1})\leqslant q^m+1$,
and $r\geqslant3$ is odd with $d+r\leqslant N$.
For every $\boldsymbol{u}\in\mathcal I$ of total dimension $d+r$,
the following bounds hold.
\begin{enumerate}[label=(\arabic*)]
\item If $q\geqslant3$, then
\begin{equation}\label{eq:odd-layer-q-ge-3}
 \omega_d(\boldsymbol{u})\geqslant(q^m-1)\frac{r-1}{r}
 \bigl(1+(q-3)q^{m-1}\bigr)\binom{d+r-3}{r-1}.
\end{equation}
Equality holds precisely when $m=1$, $d=q$, $r=3$, and all
block lengths are $1$.
\item If $q=2$ and $m\geqslant3$, then
\[
 \omega_d(\boldsymbol{u})\geqslant
 (2^m-1)\frac{r+2}{r}\binom{d+r-3}{r-1}.
\]
\end{enumerate}
\end{proposition}

\begin{proof}
Put $s=d+r$, $W=X_{s-2}$,
$x=\alpha_{s-1}/q^m$, and $y=\alpha_s/q^m$. Thus
$\lvert W\rvert=d+r-3$ and every entry of $W$ is at least $1/q^m$.
By \eqref{eq:er-gap},
\begin{equation}\label{eq:odd-refinement-gap}
 e_1(W)e_{r-1}(W)-r e_r(W)
 \geqslant\frac{r-1}{q^m}e_{r-1}(W).
\end{equation}
Also $C_{r-2}(W)\geqslant0$ under the single test, by
Theorem~\ref{thm:main} and \eqref{eq:normalized}.

Suppose first that $q\geqslant3$. Since $e_1(B)\leqslant1$ and
all normalized entries are at most $1/q$,
\[
 e_1(W)\leqslant1+\frac{r-3}{q},\qquad
 x+y\leqslant\frac2q.
\]
Consequently,
\[
 r(1-x-y)-e_1(W)\geqslant(r-1)\left(1-\frac3q\right).
\]
Adding this inequality, multiplied by $e_{r-1}(W)$, to
\eqref{eq:odd-refinement-gap} gives
\[
\begin{aligned}
 &(1-x-y)e_{r-1}(W)-e_r(W)\\
 &\quad\geqslant\frac{r-1}{r}
 \left(1-\frac3q+\frac{1}{q^m}\right)e_{r-1}(W).
\end{aligned}
\]
Using \eqref{eq:two-factor}, the nonnegativity of $C_{r-2}(W)$,
and the bound
\[
 e_{r-1}(W)\geqslant\binom{d+r-3}{r-1}q^{-m(r-1)},
\]
we obtain \eqref{eq:odd-layer-q-ge-3} from
\eqref{eq:normalized} and Lemma~\ref{lem:greedy}.

If equality holds for some $\boldsymbol{u}$, then it also holds
at $\boldsymbol{g}_s$ by Lemma~\ref{lem:greedy} and the bound
just proved. Thus $C_{r-2}(W)=0$, because
$(1-x)(1-y)>0$. Corollary~\ref{cor:strict} forces $r=3$;
hence $W=B$ and $e_1(B)=1$. Equality in the lower bound
for $e_2(B)$ forces every entry of $B$ to equal $1/q^m$.
Thus every block has length $1$, and $d=q^m$. Equality in
$x+y\leqslant2/q$ now requires $1/q^m=1/q$, so $m=1$
and $d=q$. Conversely, for these parameters the coefficient
formula gives
$\omega_q(\boldsymbol{u})=q(q-1)^2(q-2)/3$,
which equals the right-hand side of \eqref{eq:odd-layer-q-ge-3}.

Now suppose $q=2$ and $m\geqslant3$. First assume that $m\geqslant4$,
or that $m=3$ and none of $\alpha_{d+2}/8,\ldots,\alpha_s/8$ equals
$1/2$. In the first case, Lemma~\ref{lem:binary-odd} gives
$r(1-x-y)-e_1(W)\geqslant3/q^m$. In the second, the entries after
$B$ are at most $1/4$, so $e_1(W)\leqslant(r+1)/4$ and
$x+y\leqslant1/2$. Hence
\[
 r(1-x-y)-e_1(W)\geqslant\frac{r-1}{4}\geqslant\frac38.
\]
In either case, \eqref{eq:odd-refinement-gap} yields
\begin{equation}\label{eq:binary-refinement-step}
 (1-x-y)e_{r-1}(W)-e_r(W)
 \geqslant\frac{r+2}{rq^m}e_{r-1}(W).
\end{equation}
Combining \eqref{eq:binary-refinement-step} with
\eqref{eq:two-factor} gives
\[
 C_r(X_s)\geqslant\frac{r+2}{rq^m}e_{r-1}(W)
 \geqslant\frac{r+2}{rq^{mr}}\binom{d+r-3}{r-1}.
\]
The asserted bound follows from \eqref{eq:normalized} and
Lemma~\ref{lem:greedy}.

It remains to consider $m=3$ when one of
$\alpha_{d+2}/8,\ldots,\alpha_s/8$ equals $1/2$.
Lemma~\ref{lem:half-entry} gives $d=3$ and $n_1=3$.
Thus $W$ has $r$ entries, one of which is $\alpha_3/8=1/2$,
while all the others are at least $1/8$. Of the $r$ products
in $e_{r-1}(W)$, exactly $r-1$ contain this entry. Therefore
\[
 e_{r-1}(W)\geqslant(4r-3)8^{-(r-1)}
 \geqslant3r\,8^{-(r-1)}.
\]
Every adjacent pair satisfies $x+y\leqslant3/4$, and
Lemma~\ref{lem:complement} gives
$e_r(W)\leqslant e_{r-1}(W)/(2r)\leqslant e_{r-1}(W)/6$.
Equation~\eqref{eq:two-factor} consequently implies
\[
 C_r(X_s)\geqslant\frac1{12}e_{r-1}(W)
 \geqslant\frac{2r}{8^r}\geqslant\frac{r+2}{8^r}.
\]
Since $d=3$, this is the required bound after applying
\eqref{eq:normalized} and Lemma~\ref{lem:greedy}.
\end{proof}

The following bound holds for every $q$.

\begin{corollary}\label{cor:odd-bound}
Suppose $d\geqslant3$ and $S(\boldsymbol{g}_{d+1})\leqslant q^m+1$. For every odd
$r\geqslant3$ with $d+r\leqslant N$ and every $\boldsymbol{u}\in \mathcal I$
of total dimension $d+r$,
\begin{equation}\label{eq:odd-bound}
 \omega_d(\boldsymbol{u})\geqslant (q^m-1)\frac{r-1}{r}
       \binom{d+r-3}{r-1}.
\end{equation}
Equality holds if and only if $q=3$, $m=1$, $d=r=3$,
and $n_1=\cdots=n_t=1$ with $t\geqslant6$.
\end{corollary}

\begin{proof}
For $q\geqslant3$, Proposition~\ref{prop:binary-odd-refinement} gives
the bound because $1+(q-3)q^{m-1}\geqslant1$. Equality requires
$q=3$, and the equality conditions in that proposition are then
exactly those stated here.

Suppose $q=2$. The case $m=1$ is impossible: all blocks would have
length $1$, and the single test would give $d+1\leqslant3$,
contrary to $d\geqslant3$. For $m\geqslant3$, the binary bound in
Proposition~\ref{prop:binary-odd-refinement} is strictly larger than
the right-hand side of~\eqref{eq:odd-bound}, since $r+2>r-1$.

It remains to consider $m=2$. Lemma~\ref{lem:half-entry} excludes
an entry $1/2$ after the first $d+1$ positions, since such an entry
would force $d\leqslant2$. Set $s=d+r$, $W=X_{s-2}$,
$x=\alpha_{s-1}/4$, and $y=\alpha_s/4$. The entries after $B$
are therefore at most $1/4$, and
\[
 e_1(W)\leqslant1+\frac{r-3}{4},\qquad x+y\leqslant\frac12.
\]
Thus $r(1-x-y)-e_1(W)\geqslant(r-1)/4$.
Every entry of $W$ is at least $1/4$, so~\eqref{eq:er-gap} gives
\[
 e_1(W)e_{r-1}(W)-r e_r(W)
 \geqslant\frac{r-1}{4}e_{r-1}(W).
\]
Combining the last two inequalities yields
\[
 (1-x-y)e_{r-1}(W)-e_r(W)
 \geqslant\frac{r-1}{2r}e_{r-1}(W).
\]
By Theorem~\ref{thm:main} and~\eqref{eq:normalized},
$C_{r-2}(W)\geqslant0$. Since $\lvert W\rvert=d+r-3$, identity
\eqref{eq:two-factor} now gives
\[
 C_r(X_s)\geqslant\frac{r-1}{2r}e_{r-1}(W)
 \geqslant\frac{2(r-1)}{r\,4^r}\binom{d+r-3}{r-1}.
\]
After applying~\eqref{eq:normalized} and Lemma~\ref{lem:greedy},
this gives twice the positive right-hand side of~\eqref{eq:odd-bound}.
Hence equality is impossible when $q=2$.
\end{proof}

\subsection{Integrality, Zero Counts, and Minimizers}

\begin{lemma}\label{lem:integrality}
For every $\boldsymbol{u}\in\mathcal I$, the number
$\omega_d(\boldsymbol{u})/(q^m-1)$ is an integer. In particular, every
positive formal support count is at least $q^m-1$.
\end{lemma}

\begin{proof}
The definition and the integrality of $q$-binomial coefficients give
\[
 \frac{\omega_d(\boldsymbol{u})}{q^m-1}
 =\sum_{\ell=d}^{\lvert \boldsymbol{u}\rvert}\sum_{a=0}^{\ell-d}q^{ma} f_\ell(\boldsymbol{u})
 \in\mathbb Z.\qedhere
\]
\end{proof}

We call $\boldsymbol{v}\in\mathcal I$ a \emph{feasible permutation} of
$\boldsymbol{w}\in\mathcal I$ if $\boldsymbol{v}$ is obtained by
permuting the coordinates of $\boldsymbol{w}$. The next proposition
determines the zero counts of total dimension $d+1$ and the equality
cases in Lemma~\ref{lem:greedy}.

\begin{proposition}\label{prop:minimizers}
Suppose $S(\boldsymbol{g}_{d+1})\leqslant q^m+1$, and let
$\boldsymbol{u}\in\mathcal I$.
\begin{enumerate}[label=(\arabic*)]
\item\label{it:zero-locus} If $\lvert\boldsymbol{u}\rvert=d+1$, then
$\omega_d(\boldsymbol{u})=0$ if and only if
$S(\boldsymbol{g}_{d+1})=q^m+1$ and $\boldsymbol{u}$ is a feasible
permutation of $\boldsymbol{g}_{d+1}$.
\item\label{it:min-large} If $\lvert\boldsymbol{u}\rvert=s\geqslant d+3$,
then $\omega_d(\boldsymbol{u})=\omega_d(\boldsymbol{g}_s)$ if and only if
$\boldsymbol{u}$ is a feasible permutation of $\boldsymbol{g}_s$.
\item\label{it:min-d2} Let $\lvert\boldsymbol{u}\rvert=d+2$. If
$S(\boldsymbol{g}_{d+1})<q^m+1$, the conclusion of~\ref{it:min-large}
holds with $s=d+2$. If $S(\boldsymbol{g}_{d+1})=q^m+1$, let $a$ be the
index of the last positive coordinate of $\boldsymbol{g}_{d+1}$. Then
$\omega_d(\boldsymbol{u})=\omega_d(\boldsymbol{g}_{d+2})$ if and only if
$\boldsymbol{u}$ is a feasible permutation of $\boldsymbol{g}_{d+2}$, or
$a<t$, $1\leqslant(\boldsymbol{g}_{d+1})_a<n_a$, and $\boldsymbol{u}$
is a feasible permutation of $\boldsymbol{g}_{d+1}+\mathbf e_{a+1}$.
\end{enumerate}
\end{proposition}

\begin{proof}
\ref{it:zero-locus} By Lemma~\ref{lem:coefficient},
$\omega_d(\boldsymbol{u})=(q^m-1)(q^m+1-S(\boldsymbol{u}))$, and
$S(\boldsymbol{u})\leqslant S(\boldsymbol{g}_{d+1})\leqslant q^m+1$.
By Lemma~\ref{lem:S-greedy}, the first inequality is an equality only
for feasible permutations of $\boldsymbol{g}_{d+1}$.

\ref{it:min-large} In the proof of Lemma~\ref{lem:greedy}, each
transfer changes $\omega_d$ by
$-(q^{u_i}-q^{u_j-1})\omega_d(\boldsymbol{w})$, where
$\lvert\boldsymbol{w}\rvert=s-1\geqslant d+2$. By
Corollary~\ref{cor:strict}, $\omega_d(\boldsymbol{w})>0$ unless
$(q,m,d)=(2,1,2)$, so every transfer strictly decreases $\omega_d$. In
the exceptional case $(q,m,d)=(2,1,2)$ all block lengths are $1$, so
every vector of total dimension $s$ is a feasible permutation of
$\boldsymbol{g}_s$.

\ref{it:min-d2} If $S(\boldsymbol{g}_{d+1})<q^m+1$, all counts of
total dimension $d+1$ are positive by~\ref{it:zero-locus}, so the
argument for~\ref{it:min-large} applies with $s=d+2$.

Suppose $S(\boldsymbol{g}_{d+1})=q^m+1$. For necessity, take a
minimizer $\boldsymbol{u}$ that is not a feasible permutation of
$\boldsymbol{g}_{d+2}$, and arrange its coordinates in nonincreasing
order. Let $i$ be the first unfilled coordinate and $j$ the last
positive coordinate. The transfer from $j$ to $i$ in
Lemma~\ref{lem:greedy} leaves the count unchanged by minimality.
Its coefficient $q^{u_i}-q^{u_j-1}$ is positive, so the exchange
identity gives $\omega_d(\boldsymbol{u}-\mathbf e_j)=0$.
Removing the last positive unit preserves the nonincreasing order.
By~\ref{it:zero-locus}, the resulting vector is a feasible
permutation of $\boldsymbol{g}_{d+1}$. Both vectors are
nonincreasing, so
$\boldsymbol{u}-\mathbf e_j=\boldsymbol{g}_{d+1}$.

If $u_j\geqslant2$, then $j=a$, since deleting one unit leaves the
last positive coordinate unchanged. Thus
$\boldsymbol{u}=\boldsymbol{g}_{d+1}+\mathbf e_a
=\boldsymbol{g}_{d+2}$, a contradiction. Hence $u_j=1$ and $j>a$;
the nonincreasing order forces $j=a+1$. If
$(\boldsymbol{g}_{d+1})_a=n_a$, then
$\boldsymbol{u}=\boldsymbol{g}_{d+1}+\mathbf e_{a+1}
=\boldsymbol{g}_{d+2}$, again a contradiction. Thus $a<t$,
$1\leqslant(\boldsymbol{g}_{d+1})_a<n_a$, and
$\boldsymbol{u}=\boldsymbol{g}_{d+1}+\mathbf e_{a+1}$.

Conversely, the transfer from $\boldsymbol{g}_{d+1}+\mathbf e_{a+1}$
to $\boldsymbol{g}_{d+2}=\boldsymbol{g}_{d+1}+\mathbf e_a$ leaves
$\omega_d$ unchanged, since its difference is a multiple of
$\omega_d(\boldsymbol{g}_{d+1})=0$.
\end{proof}

\subsection{Arithmetic Properties}

The equality case $S(\boldsymbol{g}_{d+1})=q^m+1$ in the single test is
arithmetically restricted.

\begin{proposition}\label{prop:equality-case}
Let $h$ be the number of positive coordinates of
$\boldsymbol{g}_{d+1}$. If $S(\boldsymbol{g}_{d+1})=q^m+1$, then
\[
 h\equiv1\pmod q,\qquad h\geqslant q+1,\qquad
 d\equiv1\pmod{q-1}.
\]
In particular, if $t\leqslant q$ or $d<q$, then
$S(\boldsymbol{g}_{d+1})\leqslant q^m+1$ implies
$S(\boldsymbol{g}_{d+1})<q^m+1$, and all formal support counts with
$\lvert \boldsymbol{u}\rvert\geqslant d$ are then positive.
\end{proposition}

\begin{proof}
For $a\geqslant1$, the identity
$\genfrac{[}{]}{0pt}{}{a}{1}_q=1+q+\cdots+q^{a-1}$ shows that this
coefficient is congruent to $1$ modulo $q$ and to $a$ modulo $q-1$.
The positive coordinates of $\boldsymbol{g}_{d+1}$ sum to $d+1$, so
$S(\boldsymbol{g}_{d+1})\equiv h\pmod q$ and
$S(\boldsymbol{g}_{d+1})\equiv d+1\pmod{q-1}$. Under the assumed
equality $S(\boldsymbol{g}_{d+1})=q^m+1$, these give
$h\equiv1\pmod q$ and $d\equiv1\pmod{q-1}$.
A single positive coordinate contributes at most
$\genfrac{[}{]}{0pt}{}{m}{1}_q<q^m+1$, so $h>1$. Hence
$h\geqslant q+1$.

If $t\leqslant q$, then $h\leqslant t\leqslant q$; if $d<q$, then
$h\leqslant d+1\leqslant q$. Thus equality cannot hold in either
case. When the single test holds, its inequality is therefore strict.
Corollary~\ref{cor:strict} now gives positivity of every count with
$\lvert\boldsymbol{u}\rvert\geqslant d$. Its exceptional case
$(q,m,d)=(2,1,2)$ is excluded, since then $d=q$ and
$t=N\geqslant d+1>q$.
\end{proof}

The formal support counts also satisfy congruences.

\begin{proposition}\label{prop:congruences}
Let $\boldsymbol{u}\in\mathcal I$ with
$s=\lvert \boldsymbol{u}\rvert\geqslant d$. Write $h(\boldsymbol{u})$
for the number of positive coordinates of $\boldsymbol{u}$, and
interpret an ordinary binomial coefficient with lower index out of
range as zero. Then
\[
 \frac{\omega_d(\boldsymbol{u})}{q^m-1}\equiv
 (-1)^{s-d}\binom{h(\boldsymbol{u})-1}{s-d}\pmod q.
\]
If moreover $d\geqslant2$, then
\[
 \frac{\omega_d(\boldsymbol{u})}{q^m-1}\equiv
 (-1)^{s-d}\binom{s-2}{d-2}\pmod{q-1}.
\]
Consequently, for $d\geqslant2$ and any $\boldsymbol{u},\boldsymbol{v}\in \mathcal I$ of
the same total dimension $s\geqslant d$,
\[
 \omega_d(\boldsymbol{u})-\omega_d(\boldsymbol{v})\in(q-1)(q^m-1)\mathbb Z,
\]
and if $\boldsymbol{u}$ and $\boldsymbol{v}$ have in addition the same
number of positive coordinates, their counts are congruent modulo
$q(q-1)(q^m-1)$.
\end{proposition}

\begin{proof}
Reduce the integer-coefficient power series in~\eqref{eq:coefficient}.
Modulo $q$, the factor $1-q^jz$ becomes $1$ for $j\geqslant1$,
so $P_{\boldsymbol{u}}(z)\equiv(1-z)^{h(\boldsymbol{u})}$.
Since $s\geqslant d\geqslant1$ gives $h(\boldsymbol{u})\geqslant1$,
and $q^m\equiv0$, the fraction reduces to
$(1-z)^{h(\boldsymbol{u})-1}$.
For $d\geqslant2$, reduction modulo $q-1$ instead sends every
$1-q^jz$ and $1-q^mz$ to $1-z$; hence the fraction reduces to
$(1-z)^{s-2}$. Coefficients of $z^{s-d}$ give the two congruences.

The second congruence depends only on $s$ and $d$, so equal total
dimensions yield divisibility by $(q-1)(q^m-1)$. When the numbers
of positive coordinates also agree, the first congruence gives
the factor $q$ as well. Since $q$ and $q-1$ are coprime, the final
modulus follows.
\end{proof}

\begin{corollary}\label{cor:layer-gap}
Suppose $S(\boldsymbol{g}_{d+1})\leqslant q^m+1$. For every
$\boldsymbol{u}\in\mathcal I$ with $\lvert\boldsymbol{u}\rvert=d+1$
that is not a feasible permutation of $\boldsymbol{g}_{d+1}$,
\[
 \omega_d(\boldsymbol{u})\geqslant (q^m-1)
 \bigl(q^m+1-S(\boldsymbol{g}_{d+1})+q-1\bigr).
\]
The term $q-1$ cannot be improved when
$S(\boldsymbol{g}_{d+1})=q^m+1$.
\end{corollary}

\begin{proof}
By Lemma~\ref{lem:coefficient},
$\omega_d(\boldsymbol{u})-\omega_d(\boldsymbol{g}_{d+1})
=(q^m-1)(S(\boldsymbol{g}_{d+1})-S(\boldsymbol{u}))$. Since
$S(\boldsymbol{v})\equiv\lvert\boldsymbol{v}\rvert\pmod{q-1}$ for every
$\boldsymbol{v}\in\mathcal I$, this difference is a multiple of
$(q-1)(q^m-1)$; it is positive by Lemma~\ref{lem:S-greedy}, because
$\boldsymbol{u}$ is not a feasible permutation of
$\boldsymbol{g}_{d+1}$. This gives the bound. For
$q=3$, $m=2$, $(n_1,\ldots,n_5)=(2,2,1,1,1)$, $d=5$, and
$\boldsymbol{u}=(2,1,1,1,1)$, the right side and
$\omega_d(\boldsymbol{u})$ both equal $16$.
\end{proof}


\end{document}